\documentclass{article}

\usepackage{arxiv}

\usepackage{amsmath}
\usepackage[utf8]{inputenc} 
\usepackage[T1]{fontenc} 
\usepackage{booktabs} 
\usepackage{amsfonts}  
\usepackage{nicefrac} 
\usepackage{microtype}  
\usepackage{graphicx}
\usepackage{cite}
\usepackage{doi}
\usepackage{svg}

\newcommand{\vc}[1]{{\mathbf{ #1}}}
\newcommand{\ep}{{\mathbb {E}}}

\newcommand{\olW}{\overline{W}}

\newcommand{\LL}{{\mathcal{L}}}
\newcommand{\SINR}{{\rm SINR}}
\newtheorem{theorem}{Theorem}{}
{}
\newtheorem{remark}{Remark}{}
\newtheorem{prop}{Proposition}{}
\newtheorem{proof}{Proof}{}

\usepackage{tikz}
\usetikzlibrary{shapes, arrows, shadows, fit, patterns}

\title{Bounds on Eavesdropper Performance for
a MIMO-NOMA Downlink Scheme}

\author{ \href{https://orcid.org/0000-0001-9956-2801}{\includegraphics[scale=0.06]{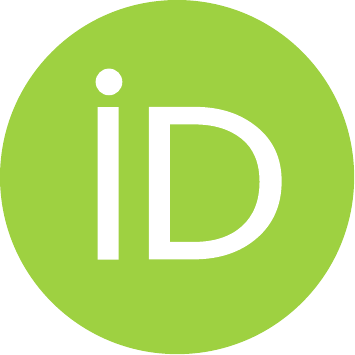}\hspace{1mm}Jennifer Chakravarty} \\
School of Mathematics\\
University of Bristol, United Kingdom\\
\texttt{jennifer.chakravarty@bristol.ac.uk} \\
\And
\href{https://orcid.org/0000-0002-3645-6670}{\includegraphics[scale=0.06]{orcid.pdf}\hspace{1mm}Oliver Johnson} \\
School of Mathematics\\
University of Bristol, United Kingdom\\
\texttt{oliver.johnson@bristol.ac.uk} \\
\AND
\href{https://orcid.org/0000-0002-4879-1206}{\includegraphics[scale=0.06]{orcid.pdf}\hspace{1mm}Robert Piechocki} \\
Department of Electrical \& Electronic Engineering \\
University of Bristol, United Kingdom \\
\texttt{robert.piechocki@bristol.ac.uk} 
}

\date{}

\renewcommand{\shorttitle}{Bounds on Eavesdropper Performance for a MIMO-NOMA Downlink Scheme}

\hypersetup{
pdftitle={Bounds on Eavesdropper Performance for a MIMO-NOMA Downlink Scheme},
pdfsubject={cs.IT},
pdfauthor={Jennifer Chakravarty, Oliver Johnson, Robert Piechocki},
pdfkeywords={NOMA, MIMO, Physical Layer Security},
}
\pgfdeclarelayer{bg}    
\pgfsetlayers{bg,main}  
\tikzstyle{decision} = [diamond, draw, fill=blue!20, text width=4.5em, text badly centered, node distance=3cm, inner sep=0pt]
\tikzstyle{block} = [rectangle, draw, fill=blue!20, text width=6em, text centered, rounded corners, minimum height=4em]
\tikzstyle{line} = [draw, -latex']
\tikzstyle{cloud} = [draw, ellipse,fill=red!20, node distance=3cm, minimum height=2em]
\tikzstyle{clear} = [rectangle,fill=white, node distance=3cm, minimum height=2em, text width=5em, text centered]  

\begin{document}
\maketitle
\begin{abstract}
Non-Orthogonal Multiple Access (NOMA) is a multiplexing technique for future wireless, which when combined with Multiple-Input Multiple-Output (MIMO) unlocks higher capacities for systems where users have varying channel strength. NOMA utilises the channel differences to increase the throughput, while MIMO exploits the additional degrees of freedom (DoF) to enhance this. This work analyses the secrecy capacity, demonstrating the robustness of a combined MIMO-NOMA scheme at physical layer, when in the presence of a passive eavesdropper. We present bounds on the eavesdropper performance and show heuristically that, as the number of users and antennas increases, the eavesdropper's SINR becomes small, regardless of how `lucky' they may be with their channel.
\end{abstract}


\section{Introduction}\label{sec:intro}
In this paper we prove that from the viewpoint of Physical Layer Security \cite{poor2017survey} the Multiple-Input Multiple-Output (MIMO) Non-Orthogonal Multiple Access (NOMA) scheme \cite{ding} protects its messages from eavesdroppers. Further, from a Massive MIMO viewpoint, as the numbers of users and antennas grow, the job of an eavesdropper becomes harder, and thus the security of the system is further enhanced. These are promising results for the inherent security of MIMO-NOMA in 5G and future wireless deployments.

NOMA~\cite{saito2013noma} is a multiplexing technique in the code or power domain, which is particularly useful when users have very different channels and path loss characteristics. 
In this framework, the base station transmits a linear combination of messages which allocates more power to the user with the weaker channel. The receivers commonly use Successive Interference Cancellation (SIC) to retrieve their signal.

NOMA is an enabling technology for 5G new radio~\cite{ding_noma_survey, 6G}, due to the performance gains obtained. Since 5G and 6G also use MIMO and Massive MIMO technology~\cite{5Gandrews, 6G}, it is  natural to ask whether MIMO and NOMA can be combined to deliver enhanced throughput relative to either scheme acting alone.

Indeed, these techniques were successfully combined by a multi-user MIMO-NOMA scheme proposed by Ding, Schober and Poor \cite{ding}, which has attracted considerable attention. The scheme of \cite{ding} was based on transmitting a linear combination of messages, mixed using a precoding matrix $P$. This $P$ is carefully designed in terms of the row spaces of the downlink channel matrices, in order to achieve signal alignment. The key property is that, for each receiver, all but one of the interfering messages are aligned in the same vector subspace, and so can be removed by projection into an orthogonal space, effectively reducing the system to a standard two-user NOMA situation. Section \ref{sec:setup} gives more details.

We consider the scheme of \cite{ding} from the point of view of an eavesdropper, in the sense of Wyner's wiretap channel as seen in Figure \ref{fig:wiretap} \cite{wyner}. Owing to the inherent randomness of the wireless medium,
we will assume that an eavesdropper has a randomly chosen channel, independent of the legitimate channel. As a result, the eavesdropper is extremely unlikely to see the same signal alignment that is achieved for the legitimate receiver. Hence, unlike the legitimate receivers, an eavesdropper cannot easily remove interfering messages meant for other receivers, and will see an inherently noisier channel.

\begin{figure}[b]
\centering
    \begin{tikzpicture}[node distance = 2cm]
        \node [block] (channelm) {Main Channel};
        \node [cloud, left of=channelm] (encoder) {Encoder};
        \node [cloud, right of=channelm] (decoder) {Decoder};
        \node [block, below of=channelm] (channele) {Eavesdropper Channel};
        \node [clear, left of=encoder] (source) {Message};
        \node [clear, right of=decoder] (outputm) {Output};
        \node [clear, right of=channele] (outpute) {Eavesdropper Ouput};
        \path [line] (source) -- (encoder);
        \path [line] (decoder) -- (outputm);
        \path [line] (channele) -- (outpute);
        \path [line] (encoder) -- (channelm);
        \path [line,dashed] (encoder) |- (channele);
        \path [line] (channelm) -- (decoder);

    \end{tikzpicture}

      \caption{The Wiretap Channel \cite{wyner}.}
       \label{fig:wiretap}

\end{figure}
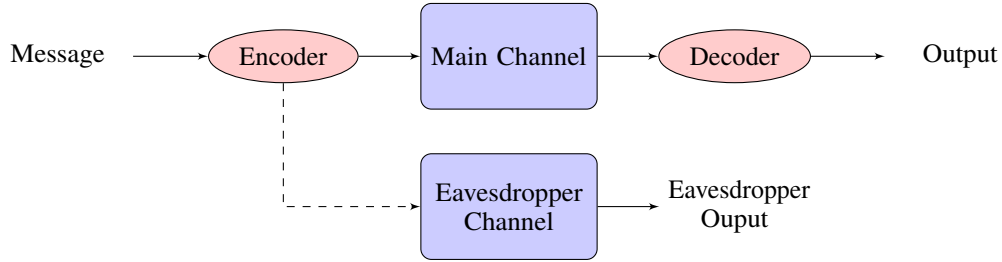

The secrecy capacity is the maximum rate at which it is possible to transmit with perfect secrecy and negligible errors. It is a natural extension of the capacity, with the additional requirement of a secure communication. It is well known that the secrecy capacity of the Gaussian wiretap channel \cite[Section 17.11]{CK_book} for a system with a single antenna transmitter, receiver and eavesdropper takes the form
\begin{equation}
    C_{s}= \max \left\{ \log(1+SINR_{M}) - \log(1+SINR_{E}), 0 \right\},
\end{equation}
where $SINR_{M}$ is the instantaneous signal-to-interference-plus-noise ratio (SINR) of the main channel and $SINR_{E}$ is the instantaneous SINR of the eavesdropper channel.
Hence when \begin{equation}SINR_{M}\geq SINR_{E},\end{equation} the secrecy capacity is non-negative and there exists a rate at which information can be sent with perfect secrecy. 

The structure of this paper is as follows:
Section \ref{sec:literature} introduces NOMA for a linear system and outlines the setup for the particular MIMO-NOMA scheme of interest in this work. Section \ref{sec:eavesdropping} considers a passive eavesdropper trying to intercept a message meant for a particular user in the system. We present the main work of this paper, providing bounds on the eavesdrop SINR and heuristic results about this SINR as the number of users increases. Section \ref{sec:large_limits} looks at limits as the number of antennas increases, representing a Massive MIMO scheme and Section \ref{sec:conclusion} concludes the work.

\section{MIMO-NOMA Systems}\label{sec:literature}

\subsection{Non-Orthogonal Multiple Access}
NOMA is a multiplexing technique, typically performed in the power domain, which was introduced by Saito et. al \cite{saito2013noma}. Users share a frequency and time slot but the power allocated to each user differs depending on their channel quality. Simply, a user with a poor channel is allocated a higher power than a user with a better channel.

In order to implement this, users with highly different channel characteristics are paired. Suppose that User 1 is closer to the base station and User 2 is further away with channel coefficients $h_{1}$ and $h_{2}$ respectively, then the base station transmits the message
$s=\alpha_{1}s_{1}+\alpha_{2}s_{2}$
where $s_i$ is the signal intended for user $i$ and the $\alpha_i$ are power allocation coefficients with $\alpha_1^2 + \alpha_{2}^{2}=1$. In this case, $\alpha_1 \leq \alpha_2$ as seen in Figure \ref{fig:NOMA_power}.

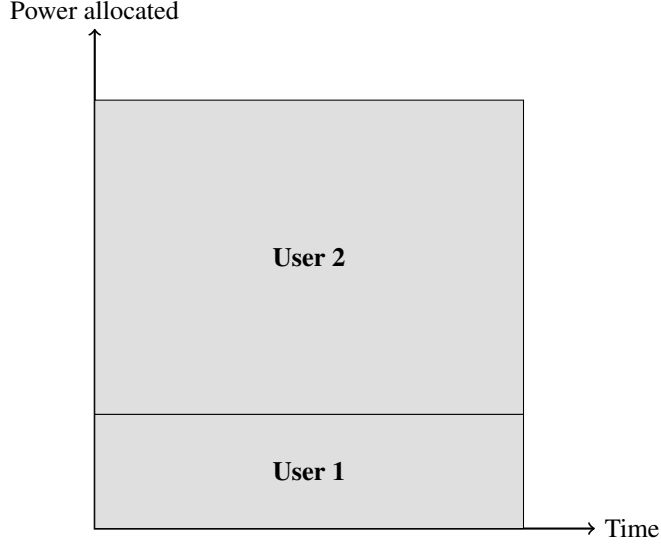
\begin{figure}[h]
    \centering
        \begin{tikzpicture}[scale=1.9]
        \draw [<->,thick] (0,3.5) node (yaxis) [above] {Power allocated}
            |- (3.5,0) node (xaxis) [right] {Time};
        \draw [draw=black, text=black, fill=lightgray!50] (0,0) rectangle (3,0.8) node[pos=.5, text=black] {\textbf{User 1}};
        \draw [draw=black, text=black, fill=lightgray!50] (0,0.8) rectangle (3,3) node[pos=.5, text=black] {\textbf{User 2}};
\end{tikzpicture}
    \caption{Power allocation in a NOMA system where User 2 has a worse channel than User 1.}
    \label{fig:NOMA_power}
\end{figure}

For $i=1,2$, user $i$ receives the message
$$y_i = h_i (\alpha_{1}s_{1}+\alpha_{2}s_{2}) + n_{i}$$
where $h_{i}$ is the channel coefficient and $n_{i}$ is noise. User 2 treats the message for User 1 as noise.
User 1 uses SIC to retrieve their message; first they find $s_2$ (which is an easier problem than for User 2, because they are closer to the base station), then they subtract this and solve for $s_1$.

\subsection{MIMO-NOMA System Setup}\label{sec:setup}
MIMO systems used in conjunction with NOMA can improve spectral efficiency \cite{DAP_2016} and thus these systems are of great interest. The key difference in a MIMO system versus a one dimensional system is that the base station now transmits a message which is a function of an information-bearing vector, where each row corresponds to the signal intended for a particular pair of receivers.


We will consider a downlink NOMA setup, and use the same notation, model and signal alignment scheme as \cite[Section II.A]{ding}. Consider a base station equipped with $M$ antennas and a collection of receivers each equipped with $N$ antennas, where $N > M/2$. The users are uniformly spaced around the base station. Assume the channel matrices from the base station to the particular users are of the form $G_m/\sqrt{L(d_m)}$ for a certain path loss function $L$ which depends on the distance $d_m$ defined as follows:
\begin{align}\label{eq:pathlossfn}
    L(d_{m})=
    \begin{cases}
    d_{m}^\alpha        &\text{if } d_{m}> r_{0}\\
    r_{0},              &\text{otherwise,}
    \end{cases}
\end{align}
for some constant $r_{0}$ and path loss exponent $\alpha$, usually between 2 and 6 for a typical 5G environment.
For brevity, we let $L_{m}$ denote $L(d_{m})$ for user $m$.

We select $M$ `near' users (within $r_1$ of the base station) and $M$ `far' users (between $r_1$ and $r_2$ from the base station) and pair them up randomly. It is required that $r_0\leq r_1$ in \eqref{eq:pathlossfn} to ensure continuity in $L$.
This setup can be seen in Figure \ref{fig:DSP_setup}.
In particular, consider pairing near users $m$ and far users $m'$ and creating a message vector $\vc{s}$ with $m$th component \begin{equation}
    \alpha_m s_m + \alpha_{m'} s_{m'},
\end{equation}
where $s_i$ is the signal intended for the $i$th user, and $\alpha_i$ are power allocation coefficients with $\alpha_m^2 + \alpha_{m'}^2 = 1$.
Since user $m'$ is further away we have that $\alpha_{m'}>\alpha_{m}$.

    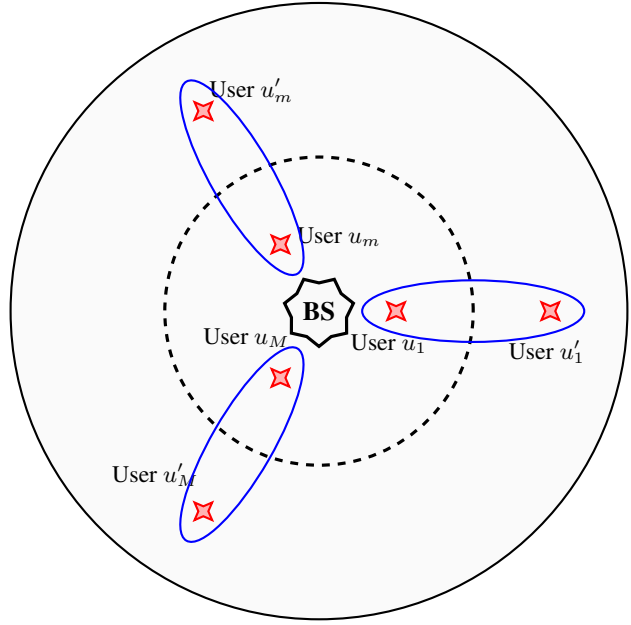
\begin{figure}[h]
        \centering
          \begin{tikzpicture}[scale=0.82,
           FIT/.style args = {#1/#2/#3}{ellipse, draw=#1, inner xsep=#2, fit=#3},
        every edge quotes/.style = {fill=white, inner sep=1pt, font=\footnotesize}]
          \begin{pgfonlayer}{bg}
                \draw[thick,fill=gray!4] (0,0) circle (5);
          \end{pgfonlayer}
          \begin{pgfonlayer}{main}
                \draw[style=dashed, very thick,fill=none] (0,0) circle (2.5);
                \node[circle] (u1) at (335:1.25cm) {\footnotesize User $u_1$};        \node[circle] (u1') at (350:3.75cm) {\footnotesize User $u_1'$};        \node[circle] (u2) at (75:1.25cm) {\footnotesize User $u_m$};        \node[circle] (u2') at (107:3.75cm) {\footnotesize User $u_m'$};        \node[circle] (um) at (200:1.25cm) {\footnotesize User $u_M$};        \node[circle] (um') at (225:3.75cm) {\footnotesize User $u_M'$};
                \node[star,star points=4,star point ratio=0.4, draw=red, thick, fill=red!30] (s1) at (0:1.25cm) {};
                \node[star,star points=4,star point ratio=0.4, draw=red, thick, fill=red!30] (s1') at (0:3.75cm) {};
                \node[star,star points=4,star point ratio=0.4, draw=red, thick, fill=red!30] (s2) at (120:1.25cm) {};
                \node[star,star points=4,star point ratio=0.4, draw=red, thick, fill=red!30] (s2') at (120:3.75cm) {};
                \node[star,star points=4,star point ratio=0.4, draw=red, thick, fill=red!30] (sm) at (240:1.25cm) {};
                \node[star,star points=4,star point ratio=0.4, draw=red, thick, fill=red!30] (sm') at (240:3.75cm) {};
                \node[star,star points=7,star point ratio=0.8, draw=black, very thick, fill=gray!5] (bs) {\textbf{BS}};
                \draw[inner ysep=-3mm,inner xsep=-2mm,rotate=-90, color=blue, thick] (0,2.5) ellipse (0.5 and 1.8);
                \draw[inner ysep=-3mm,inner xsep=-2mm,rotate=30, color=blue, thick] (0,2.5) ellipse (0.5 and 1.8);
                \draw[inner ysep=-3mm,inner xsep=-2mm,rotate=150, color=blue, thick] (0,2.5) ellipse (0.5 and 1.8);
        \end{pgfonlayer}
        \end{tikzpicture}
        \caption{User pairings in the NOMA setup based on~\cite{ding}.}
        \label{fig:DSP_setup}
    \end{figure}

The key to the scheme of \cite{ding} is the construction of  an $M \times M$ precoding matrix $P$, which is designed to make it possible to remove interference at each pair of receivers, and to reduce the problem to standard 2-user NOMA by use of an appropriate detection vector $\vc{v}$. In the framework, $\vc{v}$ is designed using signal alignment between pairs, to satisfy
    \begin{align}
        \label{eq:sig_align}\left[ G_{m}^{H}\;\;-G_{m'}^{H}\right]
        \begin{bmatrix}\vc{v_{m}}\\\vc{v_{m'}} \end{bmatrix}=\vc{0}_{M\times 1}.
    \end{align}
The columns of $P$ are designed to be orthogonal to the $m$th detection vector passed through the $m$th channel, removing the interference of other pairs of users. The existence of such a $P$ is ensured by the above signal alignment.
This is formed via constructing a matrix
$$G = [ \vc{g}_1 \; \vc{g}_2 \; \ldots \vc{g}_M]^H,$$
with $\vc{g}_m$ being a particular vector in the intersection of the row spaces of $G_m$  and $G_{m'}$ given by $\vc{g}_m^H = \vc{v}_m^H G_m$ for a certain $\vc{v}_m$. Then $$P := G^{-1} F,$$ where $F$ is a diagonal matrix chosen to ensure power constraints are met at the base station.\footnote{Note this is different to \cite[Eq. (10)]{ding} which defines $P = G^{-H} D$ for a different diagonal matrix. Since $G$ has rows $\vc{g}_i^H$, and $P$ has columns $p_j$, the necessary condition \cite[Eq. (9)]{ding}
that $\vc{g}_i^H p_j = 0$ for $i \neq j$ is achieved by taking $G P$ diagonal. Here $F= {\rm diag}(\vc{f})$ where $\vc{g}_i^H p_i = f_i$.}

The base station transmits the product $P \vc{s}$ and user $m$ receives (see \cite[Eq. (2)]{ding}):
    \begin{align}\label{eq:general_received_sig}
  \vc{y}_m & = \frac{G_{m}}{\sqrt{L_m}}  (P \vc{s}) + \vc{n} \\
& = \frac{G_{m}}{\sqrt{L_m}} \left( \sum_{i=1}^M  \left(  \alpha_m s_m + \alpha_{m'} s_{m'} \right) \vc{p}_i \right) + \vc{n}
    \end{align}
where $N \times 1$ vector $\vc{n}$ is circularly symmetric Gaussian noise with covariance proportional to $\sigma^2 \neq 0$. Note that the scheme in \cite{ding} has a factor $\rho_{I}$ denoting shot noise; for the purposes of this work we will assume there is no shot noise ($\rho_{I}$=0).

An $N\times 1$ detection vector $\vc{u}$ is applied to $\vc{y}_m$. In \cite{ding}, the choice $\vc{u} = \vc{v}_m$ is made, where the construction of the precoding matrix $P$ ensures that $\vc{v}_m^H G_m \vc{p}_i = 0$ for $i \neq m$ and $\vc{v}_m^H G_m \vc{p}_m =
\vc{g}_m^H \vc{p}_m = f_m$. This means that interference is removed and the problem is reduced to a one-dimensional NOMA problem at each receiver, with
    \begin{equation}\label{eq:receivedsig}
    y_{m} := \vc{v}^H_m \vc{y}_{m} = \frac{f_m}{\sqrt{L_{m}}}  (\alpha_m s_m + \alpha_{m'} s_{m'}) + n
    \end{equation}
where $n := \vc{v}_m^H \vc{n}$ is Gaussian noise.
Choosing the detection vector $\vc{u}$ to remove interference in this way does not necessarily maximise the SINR and thus the choice may be sub optimal, particularly in the high noise regime. We perform an analysis of the outlined system where the detection vector is designed to minimise the sum of noise and interference.

\section{Analysis of Eavesdropper Channel}\label{sec:eavesdropping}

Consider an eavesdropping receiver with an $N \times M$ channel matrix $K/\sqrt{L_{e}}$, where $K$ has IID Rayleigh elements and $L_e = L(d_e)$ applies the same path loss function $L$ to the eavesdropper distance from the base station. Without loss of generality, we will assume that the eavesdropper is listening into the message intended for User pair 1 and 1'. Since User 1' is further away, their signal receives a greater power allocation and thus will be easier to eavesdrop. We aim to show that with high probability, the eavesdropper cannot gain useful information from the message for User 1', and therefore cannot detect the message for User 1 either.

The eavesdropper receives the $N\times 1$ vector
    \begin{align}
    \vc{y}_{e}& =\frac{K}{\sqrt{L_e}}  (P \vc{s}) + \vc{n}  \nonumber \\
    & = \frac{1}{\sqrt{L_e}} \left( \sum_{i = 1}^M (\alpha_m s_m + \alpha_{m'} s_{m'}) \vc{w}_i \right) + \vc{n} \label{eq:eve_received_sig}
    \end{align}
where $N \times 1$ vector $\vc{w}_i$ is the $i$th column of $W := K P$ and the other parameters are as in \eqref{eq:general_received_sig}.

\subsection{Optimal Detection Vector}
We will consider the SINR for the eavesdropper, under the assumption that the signals $s_i$ are independent with $\ep |s_i|^2 = \rho \sigma^2$ for transmit SNR $\rho$.
Without loss of generality, we assume the eavesdropper tries to decode message $s_{1'}$ with detection vector $\vc{u}$. The eavesdropper will view all other signals as noise. The overall SINR for the communication in Equation \eqref{eq:eve_received_sig}
becomes
\begin{equation} \label{eq:SINR_eve_noma}
\SINR_{E} = \frac{ \rho |\vc{u}^{*} \vc{w}_1|^2 \alpha_{1'}^2}{\rho |\vc{u}^{*} \vc{w}_1|^2 \alpha_{1}^2 + \rho \sum_{j=2}^M |\vc{u}^{*} \vc{w}_j|^2 + L_{E} \sum_{i=1}^N |u_i|^2 }.
\end{equation}
Given the assumption that the interference noise is 0, note that this is also the SNR.

\begin{theorem}\label{thm:opt_snr_e}
    The optimal eavesdropper SINR is of the form
    \begin{align}\label{eq:eve_SINR_NOMA}
     \SINR_{E} = \frac{ \rho \alpha_{1'}^2}{\rho \alpha_{1}^2  + \left( \vc{w}_1^{*} \left( \rho (\olW \olW^{*}) + L_{E} I_N \right)^{-1} \vc{w}_1
    \right)^{-1}}.\end{align}
\end{theorem}

\begin{proof}
    
    We can find the optimal detection vector by fixing 
    \begin{align}
        \vc{u}^{*} \vc{w}_1 = \vc{w}_1^{*} \vc{u} = |\vc{u}^{*} \vc{w}_1|^2 = 1
    \end{align}
    and looking to minimise 
    \begin{align}\label{eq:optdetection_constraint}
        \rho \sum_{j=2}^M |\vc{u}^{*} \vc{w}_j|^2 + L_{E} \sum_{i=1}^N |u_i|^2.\end{align}
    The first term may be rewritten as $\rho$ multiplied by 
    \begin{align}
        \sum_{j=2}^M 
        \left( \sum_{r=1}^M u_r^* V_{rj} \right) \left( \sum_{s=1}^M u_s V_{sj}^* \right) &= \sum_{r,s=1}^M u_r^* u_s \sum_{j=2}^M V_{rj} V_{sj}^* \\
        & =  \sum_{r,s=1}^M u_r^* u_s (\olW \olW^{*})_{rs}\\
        & = \vc{u}^{*} (\olW \olW^{*}) \vc{u},
    \end{align}
    where $\olW = W - \vc{w}_1 \otimes (1, 0, \ldots, 0)$ is the matrix $W$ with its first column set to zero. 
    
    Hence, a Lagrangian formulation gives
    \begin{align}\label{eq:lagrangian}
        \LL(\vc{u}, \lambda) = \vc{u}^{*} \left( \rho (\olW \olW^{*}) + L_{E} I_N \right) \vc{u}  - \lambda \vc{u}^{*}  \left( \vc{w}_1 \vc{w}_1^{*} \right)\vc{u},
    \end{align}
    which is a complex Hermitian quadratic form and thus we find that
    \begin{align}\label{eq:lagrangian_deriv}
        \frac{\partial \LL(\vc{u}, \lambda)}{\partial \vc{u}} 
        & = 2(\rho \olW \olW^{*} + L_{E}I_{N})\vc{u} - 2\lambda  \vc{w}_1 \vc{w}_1^{*} \vc{u} = 0.
    \end{align}
    Therefore the  SINR (Equation \eqref{eq:SINR_eve_noma}) may be written as
    \begin{align}
         \SINR =& \frac{ \rho \alpha_{1'}^2}{\rho \alpha_{1}^2 + \vc{u}^{*} \left( \rho (\olW \olW^{*}) + L_{E} I_N \right) \vc{u}^*   }\\
        =&  \frac{ \rho \alpha_{1'}^2}{\rho \alpha_{1}^2 + \lambda \vc{u}^{*} \left( \vc{w}_1 \vc{w}_1^{*} \right) \vc{u}^*}
        = \frac{ \rho \alpha_{1'}^2}{\rho \alpha_{1}^2 + \lambda }.
    \end{align}

    Since $L_{E} \neq 0$, the matrix $\left( \rho (\olW \olW^{*}) + L_{E} I_N \right)$ is invertible. Hence, after some algebraic manipulation, the result follows.
\end{proof}

\begin{remark}
Note that 
the corresponding analysis will give the optimal detection vector and SINR for the legitimate user. In general this will not coincide with the choice $\vc{u} = \vc{v}_m$ made above in the analysis of Equation \eqref{eq:general_received_sig}, since that choice removes interference potentially at the cost of increased noise, whereas our analysis considers interference and noise together.    
\end{remark}

From the point of view of Physical Layer Security, if the eavesdropper channel has smaller SINR than the legitimate channel, the true message can be protected by transmitting at the relevant rate. In order to compare the two channels, we will compare the optimal SINR in each case, though note that the expression~\cite[Eq. (15)]{ding} gives a tractable upper bound on the optimal legitimate SINR. 

\section{Bounding the Eavesdropper SINR}\label{sec:bounds}

While Equation \eqref{eq:eve_SINR_NOMA} gives a closed form expression for the optimal SINR, it is stated in terms of the random quantities $\vc{w}_1$ and $\olW$, and is therefore not particularly tractable.

Writing $Z = \rho (\olW \olW^H) + L_{e} I_N$, and $R( \cdot)$ for the Rayleigh quotient, we obtain the bound
\begin{align} 
\vc{w}_1^H \left( \rho (\olW \olW^H) + L_{e} I_N \right)^{-1} \vc{w}_1   \nonumber 
&= \vc{w}_1^H \vc{w}_1 R(Z^{-1}; \vc{w}_1) \nonumber \\
 \leq  \frac{ \vc{w}_1^H \vc{w}_1 }{\lambda_{\min}(Z)} 
& = \frac{ \vc{w}_1^H \vc{w}_1 }{\rho \lambda_{\min}(\olW \olW^H) + L_{e}} \label{eq:evalbd1} \\
&\leq \frac{ \vc{w}_1^H \vc{w}_1 }{ L_{e}}. \label{eq:evalbd2} \end{align}
Note that this gives a conservative bound, since it considers the worst case and not the average case. Direct application of Equation \eqref{eq:evalbd2} means that the SINR in Equation \eqref{eq:eve_SINR_NOMA} is bounded above by
\begin{align}
\SINR \leq \frac{\rho\alpha_{1'}^{2} \vc{w}_{1}^{H} \vc{w}_{1}}{\rho\alpha_{1}^{2} \vc{w}_{1}^{H} \vc{w}_{1} + L_{e}}
\leq \frac{\rho\alpha_{1'}^{2} EW}{\rho\alpha_{1}^{2} EW + L_{e}},
\label{eq:SINRbd}
\end{align}
where $EW$ is the expectation of $ \vc{w}_{1}^{H} \vc{w}_{1}$, and the 
second inequality follows by Jensen's inequality. We plot this result in Figure \ref{fig:matlab_plot}, which shows how eavesdropper SINR decays with distance as expected, and that (owing to lack of signal alignment) on average the eavesdropper performs worse than a legitimate receiver at the same distance.

\begin{figure}[h!]
    \centering
       \includegraphics[width=0.65\textwidth]{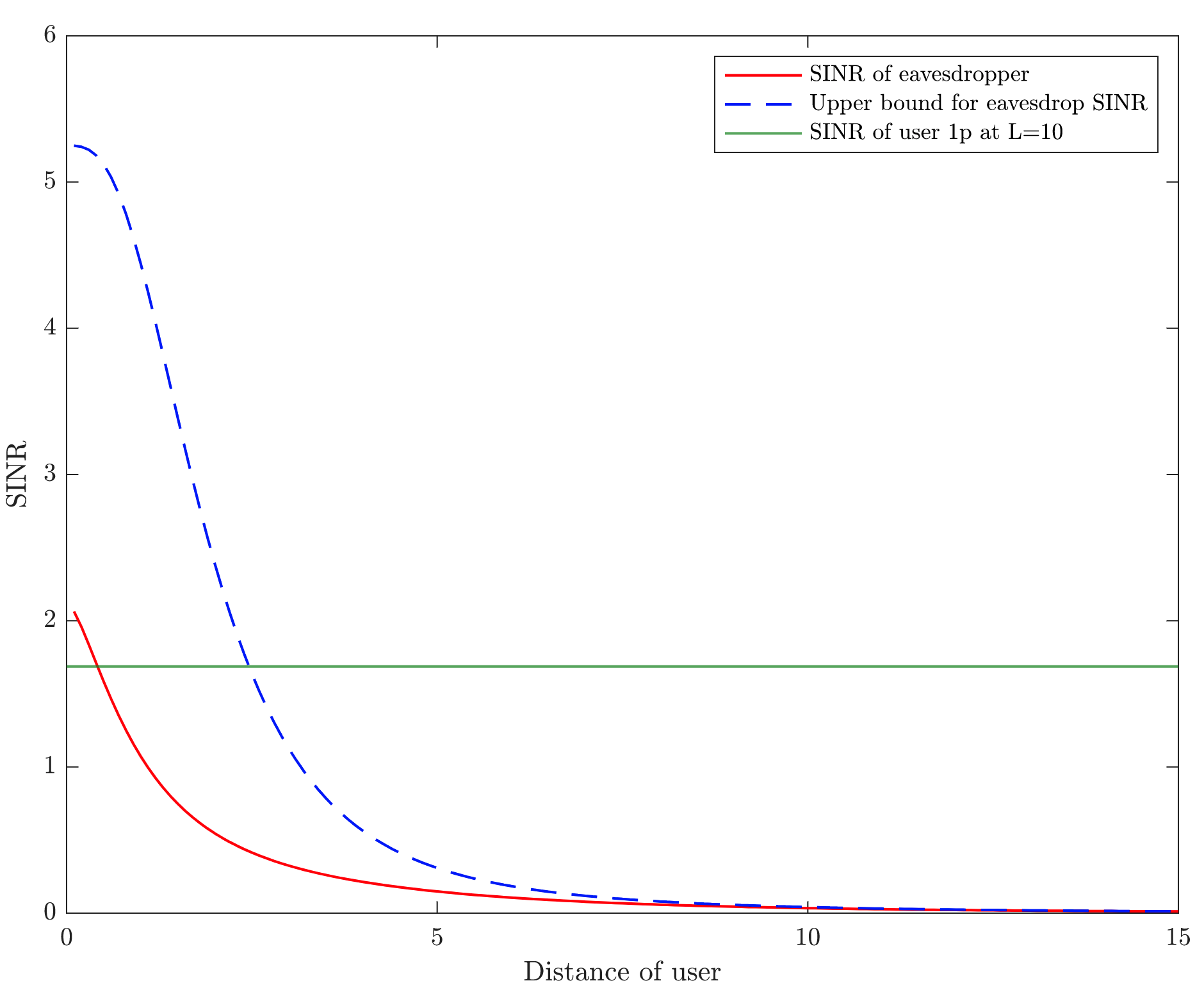}
    \caption[SINR vs distance for a MIMO-NOMA scheme]{SINR vs User distance for $M=7$, $N=5$, $\rho=5$ and legitimate users as in~\cite[Example 1]{ding}. We plot the upper bound on eavesdropper SINR from \eqref{eq:SINRbd} in blue, the empirical eavesdropper SINR from simulation in red, and the legitimate SINR in green.}
  \label{fig:matlab_plot}
\end{figure}
%


\section{Large Antenna Limits} \label{sec:large_limits}
As the number of antennas in the MIMO system are scaled up, representing a Massive MIMO NOMA setup~\cite{noma_5G}, we can argue heuristically that the eavesdropper SINR vanishes.

\begin{prop}\label{prop:noma_antenna_limits}
    In the limit of $N=\gamma M$ for $\frac{1}{2}<\gamma < 1$ then $\SINR_{E}\to 0$ at a rate of $\mathcal{O}\left(\frac{1}{M}\right)$.
\end{prop}
\begin{proof}
    Recall that $N>M/2$, so as $M$ increases so does $N$. Thus we can apply the Mar\v{c}enko--Pastur theory~\cite{marchenko}, in a regime where the number of antennas $M$ is large and $N/M \to \gamma$ (for some  $1/2 < \gamma < 1$), we have that
        \begin{align}\label{eq:MP_lambda_min}
        \lambda_{\min}(\olW \olW^{*}) \simeq c(1 - \sqrt{\gamma})^2 M
        \end{align}
    for some positive constant $c$.
    Hence for any fixed distance $L_{E}$, for $M$ sufficiently large the $\lambda_{\min}$ term will become the dominant one in Equation \eqref{eq:evalbd1} which may be estimated as
        \begin{align}
        \vc{w}_1^{*} \left( \rho (\olW \olW^{*}) + L_{E} I_N \right)^{-1} \vc{w}_1 \nonumber \\
         \leq  \frac{ \vc{w}_1^{*} \vc{w}_1 }{\rho \lambda_{\min}(\olW \olW^{*}) + L_{E}} \nonumber\\
         \simeq \frac{ \vc{w}_1^{*} \vc{w}_1 }{\rho c(1 - \sqrt{\gamma})^2 M + L_{E}}
        \end{align}
    which is a scalar value.
    Consequently, the SINR of the eavesdropper will be bounded by
        \begin{align} \label{eq:optSINR_bound}
            \SINR_{E} &= \frac{ \rho \alpha_{1'}^2}{\rho \alpha_{1}^2  + \left( \vc{w}_1^{*} \left( \rho (\olW \olW^{*}) + L_{E} I_N \right)^{-1} \vc{w}_1
            \right)^{-1}}\\
            & \leq 
            \frac{ \rho \alpha_{1'}^2}{\rho \alpha_{1}^2  + \frac{ \rho c(1 - \sqrt{\gamma})^2 M + L_{E} }{
            \vc{w}_{1}^{*} \vc{w}_{1} }
            },
        \end{align}
    which becomes arbitrarily small for large $M$.
    That is, from any position, with enough antennas and user pairs, no eavesdropping is possible.
\end{proof}

For a scenario with 50 antennas at each user, and 25 pairs of users, a result can be seen in Figure \ref{fig:50_antennas}.

\begin{figure}
    \centering
       \includegraphics[width=0.65\textwidth]{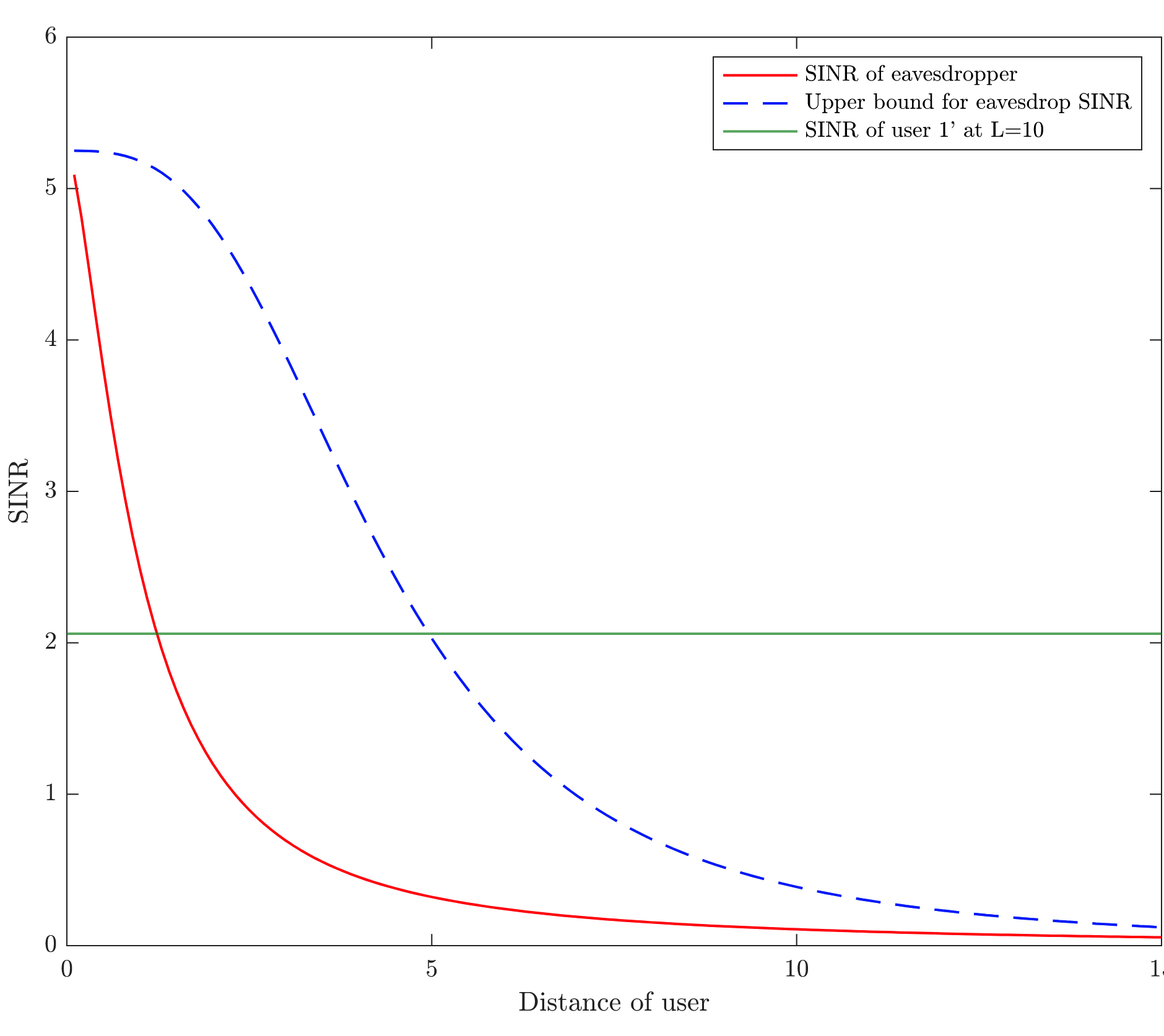}
    \caption[SINR vs distance for a MIMO-NOMA scheme]{SINR vs User distance for $M=50$ base station antennas where $\rho=10$. We plot the average empirical eavesdropper SINR from simulation in red, and the legitimate SINR in green. User 1p denotes user 1'..}
  \label{fig:50_antennas}
\end{figure}

\section{Conclusion}\label{sec:conclusion}

Schemes combining MIMO and NOMA provide great promise for the demands of 5G new radio and 6G and are likely to appear in real life systems during the upcoming years, with experimental results confirming the potential performance gains~\cite{BKY_MIMO_NOMA_trials}. While security is a key factor in modern day communication systems, it is vital to investigate their robustness to a passive eavesdropper. This work examined the combination of MIMO and NOMA in the system proposed by \cite{ding} where the message is precoded according to the legitimate user channels. This means that the message is easy to recover by a legitimate user, but difficult for users without the right channel.

It may seem that the eavesdropper could become lucky and if well aligned with the legitimate user, they could obtain the message. This work shows that as the number of user pairs increases, this is untrue and regardless of position, the SINR of the eavesdropper tends to zero with distance, meaning that they can obtain no useful information from their eavesdropping.

These results are promising for the inherent security of MIMO-NOMA systems. Since 5G, 6G and future wireless networks are densely populated with users the results in Section \ref{sec:bounds} are particularly relevant to real life architectures.

\section{Acknowledgments}\label{sec:acknowledgements}

This work was supported by the Engineering and Physical
Sciences Research Council [grant number EP/I028153/1]; GCHQ;
and the University of Bristol.

\bibliographystyle{IEEEtran}
\bibliography{bibliography.bib}

\begin{thebibliography}{10}
\providecommand{\url}[1]{#1}
\csname url@samestyle\endcsname
\providecommand{\newblock}{\relax}
\providecommand{\bibinfo}[2]{#2}
\providecommand{\BIBentrySTDinterwordspacing}{\spaceskip=0pt\relax}
\providecommand{\BIBentryALTinterwordstretchfactor}{4}
\providecommand{\BIBentryALTinterwordspacing}{\spaceskip=\fontdimen2\font plus
\BIBentryALTinterwordstretchfactor\fontdimen3\font minus
  \fontdimen4\font\relax}
\providecommand{\BIBforeignlanguage}[2]{{%
\expandafter\ifx\csname l@#1\endcsname\relax
\typeout{** WARNING: IEEEtran.bst: No hyphenation pattern has been}%
\typeout{** loaded for the language `#1'. Using the pattern for}%
\typeout{** the default language instead.}%
\else
\language=\csname l@#1\endcsname
\fi
#2}}
\providecommand{\BIBdecl}{\relax}
\BIBdecl

\bibitem{poor2017survey}
H.~V. Poor and R.~F. Schaefer, ``Wireless physical layer security,''
  \emph{Proceedings of the National Academy of Sciences}, vol. 114, no.~1, pp.
  19--26, 2017.

\bibitem{ding}
Z.~Ding, R.~Schober, and H.~V. Poor, ``A general {MIMO} framework for {NOMA
  }downlink and uplink transmission based on signal alignment,'' \emph{IEEE
  Transactions on Wireless Communications}, vol.~15, no.~6, pp. 4438--4454,
  2016.

\bibitem{saito2013noma}
Y.~Saito, Y.~Kishiyama, A.~Benjebbour, T.~Nakamura, A.~Li, and K.~Higuchi,
  ``Non-orthogonal multiple access {(NOMA)} for cellular future radio access,''
  in \emph{2013 IEEE 77th Vehicular Technology Conference (VTC Spring)}.\hskip
  1em plus 0.5em minus 0.4em\relax IEEE, 2013, pp. 1--5.

\bibitem{ding_noma_survey}
Z.~Ding, X.~Lei, G.~K. Karagiannidis, R.~Schober, J.~Yuan, and V.~K. Bhargava,
  ``A survey on {Non-Orthogonal Multiple Access} for {5G} networks: Research
  challenges and future trends,'' \emph{IEEE Journal on Selected Areas in
  Communications}, vol.~35, no.~10, pp. 2181--2195, Oct 2017.

\bibitem{6G}
W.~{Saad}, M.~{Bennis}, and M.~{Chen}, ``A {V}ision of 6{G} {W}ireless
  {S}ystems: Applications, trends, technologies, and open research problems,''
  \emph{IEEE Network}, vol.~34, no.~3, pp. 134--142, 2020.

\bibitem{5Gandrews}
J.~G. Andrews, S.~Buzzi, W.~Choi, S.~V. Hanly, A.~Lozano, A.~C.~K. Soong, and
  J.~C. Zhang, ``What will {5G} be?'' \emph{IEEE Journal on Selected Areas in
  Communications}, vol.~32, no.~6, pp. 1065--1082, June 2014.

\bibitem{wyner}
A.~D. {Wyner}, ``The wire-tap channel,'' \emph{The Bell System Technical
  Journal}, vol.~54, no.~8, pp. 1355--1387, Oct 1975.

\bibitem{CK_book}
I.~Csisz\'{a}r and J.~K\"{o}rner, \emph{Information Theory: Coding Theorems for
  Discrete Memoryless Systems}, 2nd~ed.\hskip 1em plus 0.5em minus 0.4em\relax
  Cambridge University Press, 2011.

\bibitem{DAP_2016}
Z.~{Ding}, F.~{Adachi}, and H.~V. {Poor}, ``The application of {MIMO} to
  non-orthogonal multiple access,'' \emph{IEEE Transactions on Wireless
  Communications}, vol.~15, no.~1, pp. 537--552, Jan 2016.

\bibitem{noma_5G}
L.~{Dai}, B.~{Wang}, Y.~{Yuan}, S.~{Han}, C.~{I}, and Z.~{Wang},
  ``Non-orthogonal multiple access for 5{G}: solutions, challenges,
  opportunities, and future research trends,'' \emph{IEEE Communications
  Magazine}, vol.~53, no.~9, pp. 74--81, Sep. 2015.

\bibitem{marchenko}
V.~A. Mar\v{c}enko and L.~A. Pastur, ``Distribution of eigenvalues for some
  sets of random matrices,'' \emph{Mathematics of the USSR-Sbornik}, vol.~1,
  no.~4, pp. 507--536, 1967.

\bibitem{BKY_MIMO_NOMA_trials}
A.~Benjebbour and Y.~Kishiyama, ``Combination of {NOMA} and {MIMO}: Concept and
  experimental trials,'' in \emph{2018 25th International Conference on
  Telecommunications (ICT)}, 2018, pp. 433--438.

\end{thebibliography}

\end{document}